\documentclass[prb,preprint,aps,longbibliography,groupedaddress]{revtex4-1}

\usepackage{bbm}
\usepackage{graphicx}
\usepackage{subfigure}
\usepackage{amsmath}
\usepackage{booktabs}
\usepackage{ulem}
\usepackage{xcolor}
\usepackage{pifont}
\usepackage{amssymb}
\newtheorem{lemma}{Lemma}
\newtheorem{proposition}{Proposition}

\begin{document}

\title{Intrinsic spectral structure of bipartite nonlocal magic resource}

\author{Xiao Huang$^{1}$, Guanhua Chen$^{1}$and Yao Yao$^{1,2}$\footnote{Electronic address:~\url{yaoyao2016@scut.edu.cn}}}

\address{$^1$ Department of Physics, South China University of Technology, Guangzhou 510640, China\\
$^2$ State Key Laboratory of Luminescent Materials and Devices, South China University of Technology, Guangzhou 510640, China}

\date{\today}

\begin{abstract}
Bipartite nonlocal magic resource (BNMR) quantifies the irreducible nonstabilizerness residing in bipartite entanglement, yet its evaluation is intractable due to the full Hilbert space optimization. Here, we introduce a canonical encoding framework that exactly confines the BNMR of an arbitrary bipartite pure state within a minimal encoding core. This dimension reduction proves that pure-state BNMR is an intrinsic function of the nonzero Schmidt spectrum, extending its invariance from local unitary transformations to local isometries. Leveraging this spectral link, we derive the leading quadratic response of BNMR under spectral perturbations around its zeros. We apply this quadratic response to Haar-random states, deriving and numerically validating the BNMR profile: its distribution is sharply localized at the symmetric bipartition and exponentially suppressed toward asymmetric cuts, in stark contrast to the broadening Page curve of entanglement. Finally, we provide a closed-form expression for the BNMR of Schmidt rank-2 states, uncovering a hierarchy collapse in generalized GHZ states where bipartite and global nonlocal magic resources coincide exactly.
\end{abstract}

\maketitle

\section{Introduction}
Entanglement and nonstabilizerness (magic) are two fundamental quantum resources~\cite{resource,Veitch_2014}, playing central roles in universal quantum computation ~\cite{Kitaev2005,DiVincenzo_2000,Raussendorf2001} and in the characterization of quantum many-body complexity~\cite{Gottesman1998,MPS1,ent,Cao2021}. Compared with the long-standing investigation of entanglement theory~\cite{Werner1989,PPT,Bennett1996,Quantum-entanglement}, the study of magic resource in many-body systems has developed rapidly in recent years~\cite{Oliviero2022_experimental,2025experimental}. This progress has been largely stimulated by the introduction of the Stabilizer R\'{e}nyi Entropy (SRE)~\cite{SRE,SRE2,SRE1}, together with
the numerical and experimental protocols that make SRE computable and measurable in scalable systems~\cite{flatness,Oliviero2022_experimental,2025experimental,Pauli_sampling,Pauli_Markov,magicMPS,Xiao2026exponential,Huang2026fast,Sierant_2026}. These advances have made it possible to investigate many-body phenomena from the perspective of nonstabilizerness~\cite{ManyBodymagic}, including quantum dynamics~\cite{Turkeshi_2025,magicMBL,extensive3,PXPmagic,Leone2021} and quantum phase transitions~\cite{magicMPS,Pauli_Markov,extensive4,monitor1,MIPT,long_range,Ellison2021,Oliviero2022}.

Concurrently, a crucial line of research has focused on the interplay between magic resource and entanglement~\cite{flatness,bond,haar_orbit,MIPT,Haar_scaling}. One of the central insights from this effort is that entanglement fundamentally acts as a container for magic resource~\cite{ManyBodymagic,cup,2025experimental}, motivating the concept of nonlocal magic resource~\cite{Cao21,Cao,BNM,mutual_magic,link2026,Paolo2026,munizzi2026}. Compared with entanglement alone, nonlocal magic resource captures an explicit form of quantum complexity, arising from the coupling of nonlocality and nonstabilizerness. This hierarchy is evident from the circuit-depth formulation of long-range nonstabilizerness: while long-range nonstabilizerness necessarily entails long-range entanglement, the converse is false~\cite{long_range}. Furthermore, the nonlocal magic resource encoded in generalized W-states structures has been shown to detect quantum phase transitions that remain invisible to entanglement probes~\cite{W1,W2}. Taken together, these results highlight the central role of nonlocal magic resource in characterizing quantum complexity beyond entanglement and motivate the need for its precise quantification.

Nonlocal magic resource can be formulated rigorously via local unitary (LU) optimization~\cite{Cao,BNM} to capture the nonstabilizerness surviving local basis changes. It thereby isolates the irreducible nonstabilizerness genuinely encoded within entanglement. This formulation applies naturally to different subsystem partitions, and it is inherently LU-invariant, consistent with the basic requirement that a nonlocal notion should remain stable under local basis transformations. However, the optimization formulation introduces a catastrophic computational complexity. Although a recent numerical work utilizes specific ansatz~\cite{ansatz1,ansatz2} to evaluate global nonlocal magic resource in certain many-body systems~\cite{PXPmagic}, evaluating bipartite nonlocal magic resource (BNMR) poses a steeper barrier due to the full Hilbert space optimization. Therefore, BNMR is intractable beyond a few qubits ~\cite{BNM,2qutrit}. Very recently, two independent works bypassed this bottleneck in free fermion models~\cite{NMR1,NMR2}. By restricting the optimization space to local Gaussian unitary operations, these authors derived a closed-form expression for the BNMR of pure fermionic Gaussian states, allowing them to evaluate the BNMR across various quantum many-body models.

In this work, we introduce a canonical encoding framework that compresses the nontrivial bipartite correlation structure of an arbitrary pure state into a minimal encoding core. We demonstrate that the BNMR of the original state can be exactly recovered from this core, reflecting its compact feature. This dimension reduction further establishes that pure-state BNMR is a universal function of the nonzero Schmidt spectrum, thereby extending its LU invariance to local-isometry invariance. The paper is organized as follows. The next Section gives a formal definition for the BNMR. Section III describes the canonical encoding procedure. Section IV expresses the spectral structure and dependence of BNMR, where the main propositions are proven. These propositions are applied to three typical states in Section V, where Haar-random states, Schmidt rank-2 states and generalized GHZ states are taken as prototypes. The conclusions are drawn in the final section.


\section{Stabilizer entropy and nonlocal magic resource}
We consider a quantum system of $N$ qubits. Let $\mathcal{P}_{N} = \{I, X, Y, Z\}^{\otimes N}$ denote the set of $N$-qubit Pauli operator strings, where each element is identified as $P = \prod_{j=1}^{N} P_{j}$. For a pure state $\rho$, a well-defined measure of the magic resource is the \textit{stabilizer 2-R\'{e}nyi entropy} (SRE) ~\cite{SRE}:
\begin{equation}
M_{2}(\rho) = -\ln \sum_{P \in \mathcal{P}_{N}} (1/2^{N}) \text{Tr}[\rho P]^{4}.
\label{eq:SRE}
\end{equation}
By definition, ${M_{2}}$ serves as a faithful measure that vanishes for stabilizer states and remains positive otherwise. Furthermore, it is fundamentally invariant under Clifford transformations, satisfying ${M_{2}}(U_{\rm C}\rho U_{\rm C}^{\dagger}) = {M_{2}}(\rho)$, and scales additively across composite systems, such that ${M_{2}}(\rho \otimes \rho') = {M_{2}}(\rho) + {M_{2}}(\rho')$. Throughout this work, we quantify the magic resource using the SRE at R\'{e}nyi index $2$, which is a magic resource monotone ~\cite{SRE2,SRE1}. We evaluate this quantity via the replica matrix product state (MPS) approach ~\cite{magicMPS}, which yields numerically exact results.

Nonlocal magic resource quantifies the irreducible nonstabilizerness that cannot be removed by local basis changes, thereby isolating the magic resource genuinely residing in entanglement. In this work, we focus on the BNMR defined as ~\cite{Cao,BNM}
\begin{equation}
M_{\text{bp}}(\rho_{AB}) = \min_{U= U_{A}\otimes U_{B}} M_{2}(U \rho_{AB} U^\dagger),
\label{eq:NLMR}
\end{equation}
where $A$ and $B$ denote the subsystems defined by a given bipartition. By optimizing over all LU operations, $M_{\text{bp}}$ removes the magic contributions localized within the individual subsystems. The resulting BNMR thereby captures the irreducible quantum complexity arising from the interplay between nonstabilizerness and entanglement.

As a natural $N$-partite extension, the global nonlocal magic resource captures the total amount of nonstabilizerness genuinely stored within the full entanglement network across the entire system. It is defined by minimizing the SRE over local unitary operations applied to each individual site:
\begin{equation}
M_{\text{NL}}(\rho) = \min_{U=\bigotimes_{i=1}^N U_{i}} M_{2}(U \rho U^\dagger).
\label{eq:totalNLMR}
\end{equation}

\section{Canonical encoding}
We consider an arbitrary pure state $|\Psi\rangle_{AB}$ of a many-body system with local dimension $d=2$, with a fixed cut into subsystems $A$ and $B$, containing $n_A$ and $n_B$ qubits, respectively. Assuming the Schmidt rank across this cut is $r$, the state admits the Schmidt decomposition
\begin{equation}
\label{eq:schmidt_chi}
|\Psi\rangle_{AB}
=
\sum_{\mu=1}^{r}\sqrt{\lambda_\mu}\,
|u_\mu\rangle_A \otimes |v_\mu\rangle_B,
\end{equation}
where $\lambda_\mu>0$, $\sum_{\mu=1}^{r}\lambda_\mu=1$, and
$\{|u_\mu\rangle_A\}$ and $\{|v_\mu\rangle_B\}$ are orthonormal Schmidt bases in subsystems $A$ and $B$.

Since the BNMR is invariant under local unitary operations, any state locally equivalent to $|\Psi\rangle_{AB}$ yields the same BNMR. We can choose a canonical representative that compresses the nontrivial bipartite entanglement structure within a minimal encoding sector, decoupling the remaining degrees of freedom into local pure ancillary states.

Specifically, let $k = \lceil \log_2 r \rceil$ be the minimum number of qubits required to support this Schmidt structure. We can identify a $k$-qubit encoding sector ($A_{\mathrm{code}}, B_{\mathrm{code}}$) and an ancillary sector ($A_{\mathrm{anc}}, B_{\mathrm{anc}}$) in each subsystem. There exist local unitary operations $V_A$ and $V_B$ that map the Schmidt bases into the computational bases of the encoding sectors, factoring out the ancillary qubits into the $|0\rangle$ state:
\begin{equation}
\label{eq:V_encode}
\begin{aligned}
V_A |u_\mu\rangle_A
&=
|\mu\rangle_{A_{\mathrm{code}}}
\otimes
|0\rangle^{\otimes (n_A-k)}_{A_{\mathrm{anc}}},\\
V_B |v_\mu\rangle_B
&=
|\mu\rangle_{B_{\mathrm{code}}}
\otimes
|0\rangle^{\otimes (n_B-k)}_{B_{\mathrm{anc}}},\\
\end{aligned}
\end{equation}
for $\mu=1,\dots,r$, where $|\mu\rangle_{A_{\mathrm{code}}}$ and $|\mu\rangle_{B_{\mathrm{code}}}$ are computational-basis states of the corresponding $k$-qubit sectors. If $r<2^k$, only the first $r$ basis states are occupied. Applying this local unitary transformation $V_A\otimes V_B$ to Eq.~\eqref{eq:schmidt_chi}, we obtain the encoded form
\begin{equation}
\label{eq:encoded_form}
(V_A\otimes V_B)|\Psi\rangle_{AB}
=
\left(
\sum_{\mu=1}^{r}\sqrt{\lambda_\mu}\,
|\mu\rangle_{A_{\mathrm{code}}}\otimes|\mu\rangle_{B_{\mathrm{code}}}
\right)
\otimes
|0\rangle^{\otimes (n_A-k)}_{A_{\mathrm{anc}}}
\otimes
|0\rangle^{\otimes (n_B-k)}_{B_{\mathrm{anc}}}.
\end{equation}

This representation naturally factorizes the state into a nontrivial bipartite core and a product of local pure ancillary states. By defining the core state on the minimal encoding sector as $|\phi_{r}\rangle_{\rm code} = \sum_{\mu=1}^{r}\sqrt{\lambda_\mu}\, |\mu\rangle_{A_{\mathrm{code}}}\otimes|\mu\rangle_{B_{\mathrm{code}}}$, and the corresponding ancillary state as $|\overline{0}\rangle_{\rm anc} = |0\rangle^{\otimes (n_A-k)}_{A_{\rm anc}} \otimes |0\rangle^{\otimes (n_B-k)}_{B_{\rm anc}}$, the canonical representative of the local-unitary equivalence class is given by $|\phi_{r}\rangle_{\rm code}\otimes|\overline{0}\rangle_{\rm anc}$, as schematically illustrated in Fig.~\ref{fig:LU}. Owing to the local-unitary invariance of the BNMR, we obtain
\begin{equation}
\label{eq:LU_invariance_chi}
M_{\mathrm{bp}}(|\Psi\rangle_{AB})
=
M_{\mathrm{bp}}
\!\left(|\phi_{r}\rangle_{\rm code}
\otimes
|\overline{0}\rangle_{\rm anc}
\right).
\end{equation}

\begin{figure}[htbp]
\centering
\fbox{\includegraphics[scale=0.65]{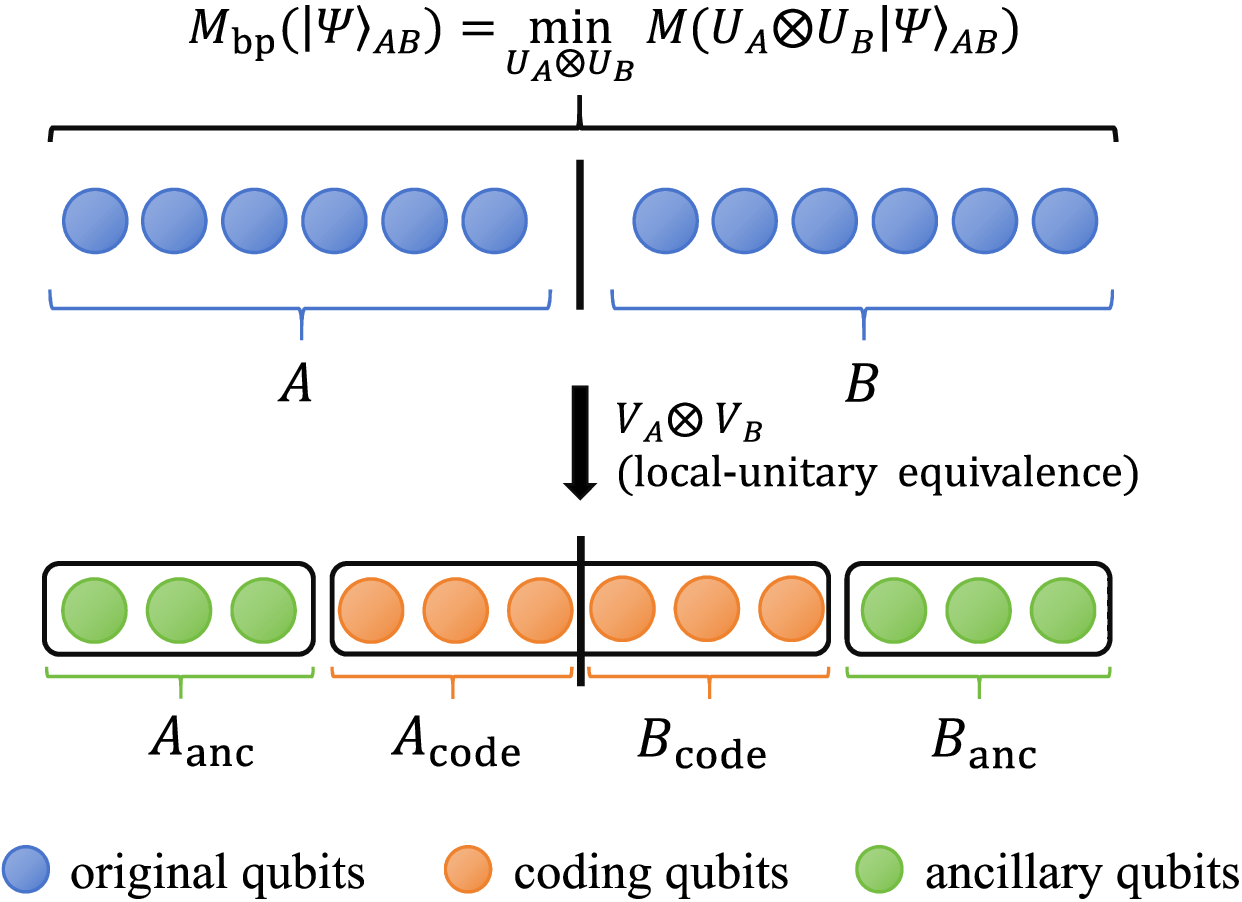}}
\caption{Schematic of the local-unitary operation reduction to the canonical encoding. An arbitrary bipartite state $|\Psi\rangle_{AB}$ (blue) is mapped to a canonical form $|\phi_{r}\rangle_{\rm code}\otimes |\overline{0}\rangle_{\rm anc}$ via $V_A \otimes V_B$. The nontrivial bipartite correlations are compressed into the minimal encoding sectors $A_{\mathrm{code}}$ and $B_{\mathrm{code}}$ (orange), decoupling the remaining degrees of freedom into local ancillary sectors $A_{\mathrm{anc}}$ and $B_{\mathrm{anc}}$ (green). }
\label{fig:LU}
\end{figure}

We now compare the BNMR of the full canonical representative, $M_{\mathrm{bp}}(|\phi_{r}\rangle_{\rm code}\otimes|\overline{0}\rangle_{\rm anc})$, with that of the minimal encoding core, $M_{\mathrm{bp}}(|\phi_{r}\rangle_{\rm code})$. Since BNMR is defined via minimization, one might mathematically suspect that expanding the optimization domain to the enlarged Hilbert space could yield a lower value than restricting it to the core sector. However, its intrinsic and compact nature excludes this possibility. The ancillary factor $|\overline{0}\rangle_{\rm anc}$ is a direct product of local pure stabilizer states. In the axiomatic frameworks of both entanglement and magic resource theories, appending or discarding local stabilizer states is a fundamental free operation~\cite{resource,LOCC,ManyBodymagic,Veitch2012}. Because BNMR captures the irreducible complexity integrating both resources, as quantified by faithful SRE, it must remain invariant under such joint free operations. Therefore, these redundant local degrees of freedom provide no nontrivial leverage to lower this intrinsic measure, dictating the equality:
\begin{equation}
\label{eq:encoding_equal_full}
M_{\mathrm{bp}}
\!\left(
|\phi_{r}\rangle_{\rm code}\otimes|\overline{0}\rangle_{\rm anc}
\right)
=
M_{\mathrm{bp}}
\!\left(
|\phi_{r}\rangle_{\rm code}
\right).
\end{equation}
Combining this with Eq.~\eqref{eq:LU_invariance_chi}, we achieve the dimension reduction:
\begin{equation}
\label{eq:final_chi}
M_{\mathrm{bp}}(|\Psi\rangle_{AB})
=
M_{\mathrm{bp}}
\!\left(
|\phi_{r}\rangle_{\rm code}
\right).
\end{equation}

Thus, for any pure many-body state under a fixed bipartition, all degrees of freedom relevant to the BNMR can be confined to a localized core comprising $k = \lceil \log_2 r \rceil$ qubits on each side. By compressing the rank-$r$ Schmidt structure into this minimal encoding sector, we bypass the global optimization across the full Hilbert space. This reduction is highly advantageous for evaluating the BNMR of low-entanglement many-body systems, including area-law ground states of local Hamiltonians ~\cite{area1,area2,area3}, as well as genuinely multipartite entangled states such as GHZ- and W-class states ~\cite{W-state,GHZ,CKW}.

\section{Spectral universality of BNMR: Local-isometry invariance}
We further show that the pure-state BNMR is a universal function of the nonzero Schmidt spectrum, extending its invariance beyond ordinary local unitary operations to arbitrary local isometries across the bipartition.

\begin{proposition} \label{prop:universal_BNMR}
For any pure bipartite state $|\Psi\rangle_{AB}$, its BNMR is an intrinsic function of the ordered nonzero Schmidt spectrum $\boldsymbol{\lambda}=(\lambda_1,\ldots,\lambda_r)$, expressed as $M_{\mathrm{bp}}(|\Psi\rangle_{AB})
=f_{\mathrm{bp}}(\boldsymbol{\lambda})$.
\end{proposition}
\noindent \textbf{Proof of Proposition 1.} Consider two bipartite pure states $|\Psi\rangle_{AB}$ and $|\Phi\rangle_{A'B'}$ with the same ordered nonzero Schmidt spectrum
$\boldsymbol{\lambda}=(\lambda_1,\ldots,\lambda_r)$. After the canonical encoding, both states can be represented by the same minimal bipartite core
\begin{equation}
\label{eq:spectral_core}
|\phi_{\boldsymbol{\lambda}}\rangle_{\rm code}
=
\sum_{\mu=1}^{r}
\sqrt{\lambda_\mu}\,
|\mu\rangle_{A_{\rm code}}\otimes|\mu\rangle_{B_{\rm code}} .
\end{equation}
Therefore, by Eq.~\eqref{eq:final_chi}, their BNMR must be identical:
\begin{equation}
\label{eq:spectral_invariance}
M_{\mathrm{bp}}(|\Psi\rangle_{AB})
=
M_{\mathrm{bp}}(|\Phi\rangle_{A'B'}) .
\end{equation}
It follows that, for pure states, the BNMR is a universal function of the nonzero Schmidt spectrum,
\begin{equation}
\label{eq:spectral_function}
M_{\mathrm{bp}}(|\Psi\rangle_{AB})
=
f_{\mathrm{bp}}(\boldsymbol{\lambda}) .
\end{equation}

This result demonstrates that BNMR is fundamentally determined by the entanglement structure. Crucially, since the proof places no restrictions on the respective Hilbert space dimensions of $|\Psi\rangle_{AB}$ and $|\Phi\rangle_{A'B'}$, it establishes that pure-state BNMR is inherently independent of the overall system size. This independence formally extends the established local-unitary invariance of BNMR to a broader local-isometry invariance, providing microscopic support for the recent view that entanglement structures can host magic resources~\cite{cup,Huang}. Note that $f_{\mathrm{bp}}$ should be regarded as a spectral map rather than as an explicit closed-form function, since evaluating $f_{\mathrm{bp}}(\boldsymbol{\lambda})$ generally requires solving a nontrivial optimization problem ~\cite{BNM,2qutrit}.

While both BNMR and entanglement entropy are functions of $\boldsymbol{\lambda}$, they characterize fundamentally different spectral features. This distinction is clearly illustrated in two-qubit pure states, where the closed-form expression for BNMR exhibits a non-monotonic dependence on the majorization order of the Schmidt spectrum ~\cite{BNM}. Since this order determines pure-state LOCC convertibility in the two-qubit case, this behavior rules out BNMR as an LOCC monotone.

While previous studies primarily connected bipartite entanglement to the total magic resource of the full system, we reveal that the genuine connection lies between entanglement and BNMR, with the Schmidt spectrum serving as the fundamental bridge. This structural connection provides a precise microscopic mechanism for macroscopic observations, such as the link between total magic resource and entanglement spectrum flatness averaged over a Clifford orbit~\cite{flatness}. Viewed through our framework, as a state evolves along the orbit, nonlocal operations iteratively induce fluctuations in BNMR across the given bipartition. Owing to the universal functional dependence of BNMR on the Schmidt spectrum, these dynamic variations must be accompanied by corresponding structural changes in the entanglement spectrum. Consequently, the macroscopic orbit-averaged response is fundamentally a statistical manifestation of this microscopic mapping.

The spectral nature of BNMR allows its zero-resource sector to be characterized directly within the Schmidt simplex. Cao \textit{et al.} proved that the zero set of BNMR corresponds exactly to the set of stabilizer-compatible flat spectra ~\cite{Cao}.
\begin{lemma}\label{lemma:zero_BNMR}
A pure bipartite state $|\Psi\rangle_{AB}$ has vanishing BNMR if and only if its Schmidt spectrum $\boldsymbol{\lambda}$ is flat and stabilizer-compatible. Here, flatness means that all nonzero Schmidt eigenvalues are equal, while stabilizer compatibility requires the Schmidt rank to satisfy $r=2^n$ for some $n\in\mathbb{Z}_{\geq 0}$.
\end{lemma}
This result follows directly from the faithfulness of pure-state magic resource measures ~\cite{SRE1,BNM} and the local equivalence of bipartite stabilizer states to $q$ independent Bell pairs ~\cite{stabilizer}. It establishes a strict boundary for the zero-resource sector: any pure state whose Schmidt spectrum deviates from perfect flatness, or whose rank is not a power of two, necessarily carries non-zero BNMR.

Having established this zero-resource boundary, a natural question is how BNMR scales under small deviations from stabilizer-compatible flat spectra. While Cao \textit{et al.} derived the leading-order expansion of an upper bound for BNMR in this near-flat regime ~\cite{Cao}, the corresponding behavior of BNMR itself has remained elusive. In Proposition~\ref{prop:perturbative_BNMR}, we close this gap by deriving its matching perturbative response, thereby proving that the existing upper bound is leading-order tight in the fixed-rank near-flat regime.

\begin{proposition} \label{prop:perturbative_BNMR}
Let $|\Psi\rangle_{AB}$ be a pure bipartite state with a fixed Schmidt rank $r=2^n$, where $n\geq 1$. Suppose that its nonzero Schmidt spectrum $\boldsymbol{\lambda}$ is a perturbation of the stabilizer-compatible flat spectrum $\boldsymbol{u}_r = (1/r, \dots, 1/r)$, namely, $\boldsymbol{\lambda} = \boldsymbol{u}_r + \boldsymbol{\eta}$. Here, $\boldsymbol{\eta} = (\eta_1, \dots, \eta_r)\in\mathbb{R}^r$ satisfies $\sum_{i=1}^r \eta_i = 0$. Then, as $\|\boldsymbol{\eta}\|_2\to0$ at fixed $r$ and fixed local Hilbert-space dimensions,
\begin{equation} \label{eq:BNMR_expansion}
    M_{\mathrm{bp}}(\boldsymbol{\lambda})
    = r \|\boldsymbol{\eta}\|_2^2 + O(\|\boldsymbol{\eta}\|_2^3).
\end{equation}
\end{proposition}

To complete the proof of Proposition~\ref{prop:perturbative_BNMR}, we utilize the following two lemmas, whose proofs are provided in the Appendix.
\begin{lemma}\label{lemma:local_expansion}
Let $|s_0\rangle \in \mathrm{Stab}$ be a fixed stabilizer state. For any normalized pure state $|\psi\rangle$ sufficiently close to $|s_0\rangle$,
\begin{equation} \label{eq:local expansion}
M_2(|\psi\rangle) = 4\left(1 - |\langle s_0|\psi\rangle|^2\right) + O\bigl( (1 - |\langle s_0|\psi\rangle|^2)^2 \bigr).
\end{equation}
Define $\delta_{\mathrm{stab}}(|\psi\rangle) = 1 - \max_{|\tau\rangle \in \mathrm{Stab}} |\langle \tau|\psi\rangle|^2$ as the fidelity distance from a state $|\psi\rangle$ to the stabilizer set.
For a fixed system size, the stabilizer set $\mathrm{Stab}$ is finite. Consequently, there exist a neighborhood of $\mathrm{Stab}$ and a constant $C_1 > 0$ such that every pure state in this neighborhood satisfies
\begin{equation} \label{eq:global expansion}
M_2(|\psi\rangle) \geq 4\delta_{\mathrm{stab}}(|\psi\rangle) - C_1 \delta_{\mathrm{stab}}(|\psi\rangle)^2.
\end{equation}
\end{lemma}

\begin{lemma} \label{lemma:spectral_lower_bound}
Let $|\Psi\rangle_{AB}$ be a pure bipartite state with Schmidt rank $r=2^n$, whose Schmidt spectrum $\boldsymbol{\lambda}$ is perturbatively close to the flat spectrum $\boldsymbol{u}_r$, such that $\boldsymbol{\lambda} = \boldsymbol{u}_r + \boldsymbol{\eta}$ with $\sum_{i=1}^r \eta_i = 0$ and $\|\boldsymbol{\eta}\|_2 \ll 1$. Then, its fidelity distance to the stabilizer set, $\delta_{\mathrm{stab}}(|\Psi\rangle_{AB})$, satisfies the lower bound:
\begin{equation} \label{eq:delta_lower_bound}
    \delta_{\mathrm{stab}}(|\Psi\rangle_{AB}) \geq \frac{r}{4} \|\boldsymbol{\eta}\|_2^2 - C_2 \|\boldsymbol{\eta}\|_2^3,
\end{equation}
for some constant $C_2 > 0$.
\end{lemma}

\noindent\textbf{Proof of Proposition \ref{prop:perturbative_BNMR}.}
In the following, we establish the perturbative scaling by deriving upper and lower bounds of $M_{\mathrm{bp}}(\boldsymbol{\lambda})$ and showing they coincide at leading quadratic order.

\textbf{Upper bound.}---To obtain an upper bound on the BNMR for a pure state with a perturbed Schmidt spectrum $\boldsymbol{\lambda}$, we consider the canonical Schmidt representative
\begin{equation}\label{eq:up}
    |\psi_{\boldsymbol{\lambda}}\rangle = \sum_{i=1}^r \sqrt{\lambda_i} |i\rangle_A |i\rangle_B,
\end{equation}
where $\{|i\rangle_A\}$ and $\{|i\rangle_B\}$ denote the computational bases throughout. Its reduced states of both subsystems $A$ and $B$ reside within the stabilizer polytope, carrying vanishing local magic resource. Therefore, the SRE of this representative can serve as a natural variational upper bound on BNMR, namely $M_{\mathrm{bp}}(\boldsymbol{\lambda}) \leq M_2(|\psi_{\boldsymbol{\lambda}}\rangle)$.

For the flat spectrum $\boldsymbol{u}_r$, the corresponding reference state is
\begin{equation}\label{eq:reference}
    |\phi_r\rangle = \frac{1}{\sqrt{r}} \sum_{i=1}^r |i\rangle_A |i\rangle_B.
\end{equation}
Because the Schmidt rank is $r=2^n$, $|\phi_r\rangle$ is local-Clifford equivalent to $n$ independent Bell pairs ~\cite{stabilizer}, establishing it as a stabilizer state. The fidelity between the perturbed state and this reference stabilizer state is $F = |\langle\phi_r|\psi_{\boldsymbol{\lambda}}\rangle|^2 = \frac{1}{r}(\sum_{i=1}^r \sqrt{\lambda_i})^2$. Substituting the perturbation $\lambda_i = 1/r + \eta_i$, we expand the coefficients to second order:
\begin{equation}
    \sqrt{\lambda_i} = \frac{1}{\sqrt{r}} + \frac{\sqrt{r}}{2}\eta_i - \frac{r^{3/2}}{8}\eta_i^2 + O(\eta_i^3).
\end{equation}
Summing over $i$ and using the zero-sum constraint $\sum_{i=1}^r \eta_i = 0$, the first-order linear term vanishes. This yields the expanded fidelity:
\begin{equation}\label{eq:fidelity_expansion}
    F = 1 - \frac{r}{4}\|\boldsymbol{\eta}\|_2^2 + O(\|\boldsymbol{\eta}\|_2^3).
\end{equation}
Applying Lemma \ref{lemma:local_expansion} around the stabilizer state $|\phi_r\rangle$, we have $M_2(|\psi_{\boldsymbol{\lambda}}\rangle) = 4(1-F) + O((1-F)^2)$. Substitution of Eq.~\eqref{eq:fidelity_expansion} then leads to the upper bound:
\begin{equation} \label{eq:upper_bound}
    M_{\mathrm{bp}}(\boldsymbol{\lambda}) \leq M_2(|\psi_{\boldsymbol{\lambda}}\rangle) = r\|\boldsymbol{\eta}\|_2^2 + O(\|\boldsymbol{\eta}\|_2^3).
\end{equation}

\textbf{Lower bound.}---Since the local unitary group is compact and $M_2$ is continuous, there exist optimal local unitary operations $U_A^*$ and $U_B^*$ such that the state $|\chi^*\rangle = (U_A^* \otimes U_B^*)|\psi_{\boldsymbol{\lambda}}\rangle$ achieves the minimum, i.e., $M_{\mathrm{bp}}(\boldsymbol{\lambda}) = M_2(|\chi^*\rangle)$.

The upper bound in Eq.~\eqref{eq:upper_bound}, together with the nonnegativity of $M_2$, implies
$M_2(|\chi^*\rangle)=O(\|\boldsymbol{\eta}\|_2^2)$, and hence $M_2(|\chi^*\rangle)\to0$ as
$\|\boldsymbol{\eta}\|_2\to0$. Since $M_2$ is continuous and its zero set is precisely the stabilizer set, compactness further implies $\delta_{\mathrm{stab}}(|\chi^*\rangle)\to0$. Thus,
$|\chi^*\rangle$ lies within the uniform stabilizer neighborhood of Lemma~\ref{lemma:local_expansion} for all sufficiently small $\|\boldsymbol{\eta}\|_2$, and therefore
\begin{equation}
\label{eq:proof_lower1}
    M_{\mathrm{bp}}(\boldsymbol{\lambda})
    =
    M_2(|\chi^*\rangle)
    \geq
    4\delta_{\mathrm{stab}}(|\chi^*\rangle)
    -
    C_1\delta_{\mathrm{stab}}(|\chi^*\rangle)^2.
\end{equation}
The right side of Eq.~\eqref{eq:proof_lower1} has the form $f(\delta)=4\delta-C_1\delta^2$. Since
$f'(\delta)=4-2C_1\delta>0$ in a sufficiently small neighborhood of $\delta=0$, $f$ is monotonically increasing throughout the relevant perturbative regime. Applying Lemma~\ref{lemma:spectral_lower_bound} to $|\chi^*\rangle$ gives
\begin{equation}
\label{eq:proof_lower2}
    \delta_{\mathrm{stab}}(|\chi^*\rangle)
    \geq
    \frac{r}{4}\|\boldsymbol{\eta}\|_2^2
    -
    C_2\|\boldsymbol{\eta}\|_2^3.
\end{equation}
For sufficiently small $\|\boldsymbol{\eta}\|_2$, both sides of Eq.~\eqref{eq:proof_lower2} lie within the interval where $f$ is monotonically increasing. Combining Eqs.~\eqref{eq:proof_lower1} and
\eqref{eq:proof_lower2}, we obtain
\begin{align}
    M_{\mathrm{bp}}(\boldsymbol{\lambda})
    &\geq
    4\left(
    \frac{r}{4}\|\boldsymbol{\eta}\|_2^2
    -
    C_2\|\boldsymbol{\eta}\|_2^3
    \right)
    -
    C_1\left(
    \frac{r}{4}\|\boldsymbol{\eta}\|_2^2
    -
    C_2\|\boldsymbol{\eta}\|_2^3
    \right)^2
    \nonumber\\
    &=
    r\|\boldsymbol{\eta}\|_2^2
    -
    4C_2\|\boldsymbol{\eta}\|_2^3
    -
    O(\|\boldsymbol{\eta}\|_2^4)
    \nonumber\\
    &=
    r\|\boldsymbol{\eta}\|_2^2
    -
    O(\|\boldsymbol{\eta}\|_2^3).
\label{eq:lower_bound}
\end{align}

In consequence, the upper and lower bounds in Eqs.~\eqref{eq:upper_bound} and \eqref{eq:lower_bound} agree with each other to second order, fixing both the quadratic scaling and its leading coefficient:
\begin{equation}
    M_{\mathrm{bp}}(\boldsymbol{\lambda})
    =
    r\|\boldsymbol{\eta}\|_2^2
    +
    O(\|\boldsymbol{\eta}\|_2^3).
\end{equation}
Equivalently,
\begin{equation}
M_{\mathrm{bp}}(\boldsymbol{\lambda}) \to r\|\boldsymbol{\eta}\|_2^2 \quad \text{as } \|\boldsymbol{\eta}\|_2 \to 0.
\end{equation}
This completes the proof of Proposition~\ref{prop:perturbative_BNMR}.

\section{Application to typical states}
\subsection{Perturbative response of Haar-random states}
We now combine Proposition~\ref{prop:perturbative_BNMR} with the statistical properties of Haar-random pure states $|\Psi\rangle$ on a bipartite Hilbert space $\mathcal{H}_A \otimes \mathcal{H}_B$, with dimensions $r = 2^{n_A}$ and $R = 2^{n_B}$ ($r \leq R$). For a full-rank reduced density matrix $\rho_A = \mathrm{Tr}_B |\Psi\rangle\langle\Psi|$, its Schmidt spectrum $\boldsymbol{\lambda}$ can be parameterized as $\boldsymbol{\lambda} = \boldsymbol{u}_r + \boldsymbol{\eta}$. To simplify notation, we denote the squared 2-norm of the perturbation, introduced in Proposition~\ref{prop:perturbative_BNMR}, by $D \equiv \|\boldsymbol{\eta}\|_2^2$. This squared spectral distance $D$ is related to the purity of the reduced state:
\begin{equation}
D = \sum_{i=1}^r \left(\lambda_i - \frac{1}{r}\right)^2 = \mathrm{Tr}(\rho_A^2) - \frac{1}{r}.
\end{equation}

For a Haar-random bipartite pure state, the average purity is $\mathrm{E}_{\mathrm{Haar}}[\mathrm{Tr}(\rho_A^2)] = \frac{r+R}{rR+1}$. Hence, the average spectral distance evaluates to
\begin{equation}
\mathrm{E}_{\mathrm{Haar}}[D] = \frac{r+R}{rR+1} - \frac{1}{r} = \frac{r^2-1}{r(rR+1)}.
\end{equation}
Substituting this result into the leading-order expansion in Eq.~(\ref{eq:BNMR_expansion}) from Proposition~\ref{prop:perturbative_BNMR}, we obtain the Haar-averaged quadratic response of the BNMR induced by spectral fluctuations, which we denote by $M_{\mathrm{bp}}^{(2)}$:
\begin{equation}\label{eq:gamma}
\mathrm{E}_{\mathrm{Haar}}[M^{(2)}_{\mathrm{bp}}] = r \, \mathrm{E}_{\mathrm{Haar}}[D] = \frac{r^2-1}{rR+1}= \gamma \frac{r^2-1}{r^2+\gamma},
\end{equation}
where $\gamma=r/R\leq1$ is the aspect ratio. At fixed $\gamma$, the expected quadratic response approaches to $\gamma = 2^{n_A - n_B}$ in the large-rank limit, indicating that for fixed $n_A$ the BNMR is exponentially suppressed by the increasing environment size $n_B$, as illustrated in Fig.~\ref{fig:HaarBNMR}. To assess the analytical quadratic response, we also compare it with numerical simulations. Specifically, we generate 500 independent Haar-random states and evaluate the SRE of their canonical Schmidt representatives in Eq.~\eqref{eq:up}, denoted by $M_{\mathrm{bp}}^{\mathrm{up}}$, which provides an upper bound on the BNMR. We refer to this quantity as the canonical upper bound.

\begin{figure}[htbp]
\centering
\includegraphics[scale=0.32]{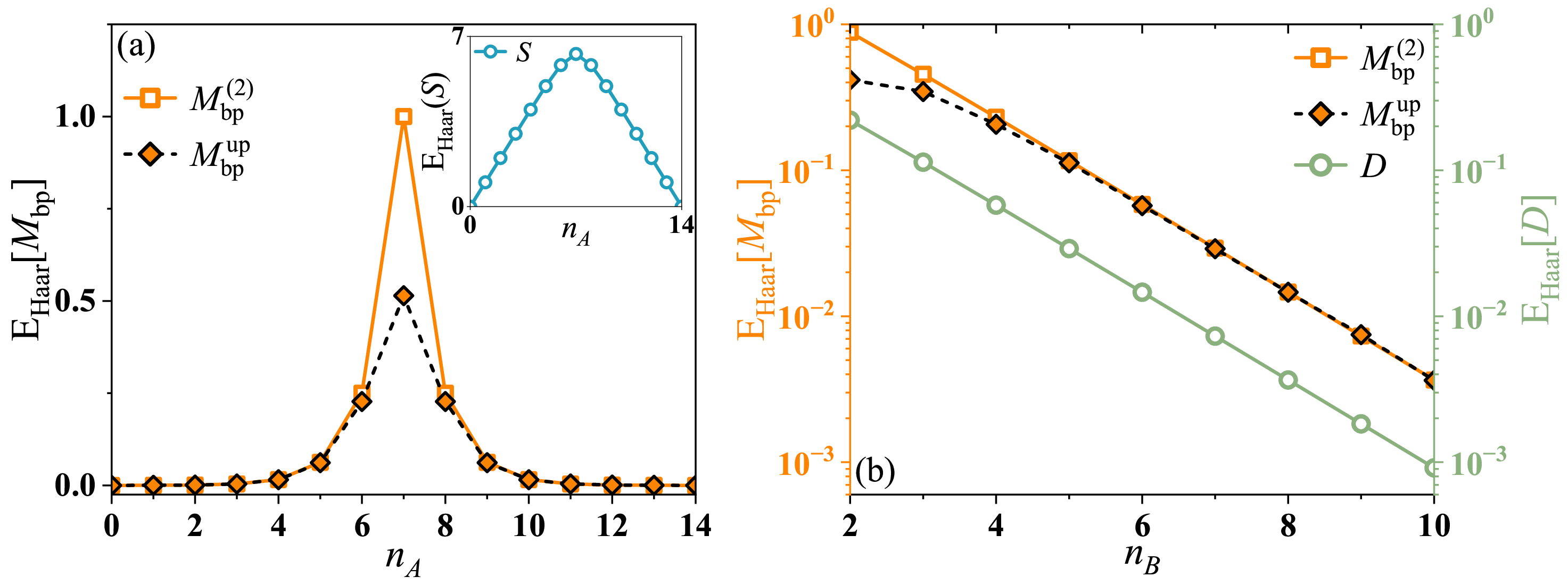}
\caption{BNMR response in the Haar-random states. Orange squares denote the analytical quadratic response $M_{\mathrm{bp}}^{(2)}$, while black diamonds denote the canonical upper bound $M_{\mathrm{bp}}^{\mathrm{up}}$ averaged over 500 independent Haar-random states. (a) Distribution across the bipartitions of a 14-qubit system. Both $M_{\mathrm{bp}}^{(2)}$ and $M_{\mathrm{bp}}^{\mathrm{up}}$ are sharply concentrated at the symmetric cut $n_A=n_B=7$, in contrast to the broadening Page curve of the entanglement entropy shown in the inset ~\cite{Page1,Page2,Page3}. The discrepancy between $M_{\mathrm{bp}}^{(2)}$ and $M_{\mathrm{bp}}^{\mathrm{up}}$ at the symmetric cut signals the breakdown of the quadratic approximation. (b) Dependence on the environment size $n_B$ at fixed $n_A=2$. The quadratic response $M_{\mathrm{bp}}^{(2)}$, the canonical upper bound $M_{\mathrm{bp}}^{\mathrm{up}}$, and spectral distance $D$ (green circles) all decay exponentially with $n_B$.}
\label{fig:HaarBNMR}
\end{figure}

To capture the global profile across different bipartitions, we examine the Haar-averaged quadratic response and canonical upper bound for a 14-qubit system, as shown in Fig.~\ref{fig:HaarBNMR}(a). Their close agreement across the perturbative regime allows them to faithfully capture the underlying BNMR distribution, confirming the tightness of the upper bound established in our proof. In sharp contrast to the extensive and broadening Page curve of the entanglement entropy as shown in the inset~\cite{Page1,Page2,Page3}, the BNMR distribution is highly localized at the symmetric bipartition ($n_A=n_B=7$, $\gamma=1$) and rapidly vanishes toward asymmetric cuts. Supported by the concentration of Haar-measure phenomenon~\cite{Mele2024haar,Turkeshi_2025}, we expect these finite-size numerical benchmarks to capture a universal BNMR profile that remains valid for larger Haar-random systems.

While the quadratic response accurately captures the BNMR behavior across asymmetric cuts, it breaks down at the symmetric bipartition. For $r=R$, Eq.~\eqref{eq:gamma} predicts $\mathrm{E}_{\mathrm{Haar}}[M_{\mathrm{bp}}^{(2)}]\to1$. However, for the 14-qubit system shown in Fig.~\ref{fig:HaarBNMR}(a), where $r=R=2^7$, the Haar-averaged canonical upper bound $\mathrm{E}_{\mathrm{Haar}}[M_{\mathrm{bp}}^{\mathrm{up}}]$ at the central cut is approximately $0.514$. This substantial discrepancy originates from the Marchenko--Pastur statistics of the Schmidt spectrum at the central bipartition of Haar-random states ~\cite{MPstatistics}. In this regime, the spectral distance scales as $D=O(r^{-1})$, while the rank factor $r$ compensates this decay, leaving the leading-order contribution $rD$ of order unity. In turn, the fidelity deviation from the flat-spectrum stabilizer reference state is no longer perturbatively small. The state therefore cannot be treated within the local expansion of Lemma~\ref{lemma:local_expansion}, and the quadratic approximation fails at the symmetric cut.

Despite the quantitative breakdown of the quadratic approximation at the symmetric cut, it retains a nontrivial prediction for the scaling toward saturation. For the symmetric bipartition ($r = R$, $\gamma=1$), the quadratic response simplifies to $\mathrm{E}_{\mathrm{Haar}}[M^{(2)}_{\mathrm{bp}}] = (r^2-1)/(r^2+1)$ and hence the residual deviation from its saturation value is
\begin{equation}
\Delta M_{\mathrm{bp}}^{(2)}(r) = 1- \mathrm{E}_{\mathrm{Haar}} [M_{\mathrm{bp}}^{(2)}(r)] = \frac{2}{r^2+1} = O(r^{-2}).
\end{equation}

Remarkably, at fixed system size $N$, this inverse-square law matches the $O(\chi^{-2})$ bond-dimension scaling of magic resource saturation in Haar-random states as recently demonstrated by Lami \textit{et al.} ~\cite{Haar_scaling}. Although the quadratic response does not reproduce the absolute BNMR value at the symmetric cut, it recovers the same algebraic exponent. By directly relating this exponent to Schmidt-spectrum fluctuations across a single bipartition, our analysis suggests that the observed $O(\chi^{-2})$ saturation scaling is not solely an artifact of the MPS ansatz ~\cite{MPS1,MPS2}, but instead reflects a general bipartite mechanism governed by BNMR.

To estimate the total BNMR in an $N$-qubit Haar-random state, we assume that the BNMR contributions captured across different spatial cuts are non-overlapping and therefore additive. We then sum the quadratic responses over all bipartitions. Since the quadratic response closely agrees with canonical upper bound across asymmetric cuts and exceeds it at the symmetric cut, we use this sum as a upper estimate of the total BNMR. For even $N$, the symmetry of the Schmidt rank $r=2^{\min(n_A,N-n_A)}$ gives
\begin{equation}
\begin{aligned}
M_N^{(\mathrm{even})} &= \sum_{n_A=1}^{N-1} \mathrm{E}_{\mathrm{Haar}} [M_{\mathrm{bp}}^{(2)}(n_A)] \\ &= \sum_{n_A=1}^{N-1} \frac{ 4^{\min(n_A,N-n_A)}-1 }{ 2^N+1 } \\ &= \frac{1}{2^N+1} \left[ 2\sum_{n_A=1}^{N/2-1} \left(4^{n_A}-1\right) + \left(4^{N/2}-1\right) \right] \\&=
\frac{1}{2^N+1}
\left[
\frac{5}{3}\left(2^N-1\right)-N
\right].
\end{aligned}
\end{equation}
Hence, $\lim_{N\to\infty}M_N^{(\mathrm{even})}=5/3$. For odd $N$, the sum converges to $\lim_{N\to\infty}M_N^{(\mathrm{odd})}=4/3$. The total magic resource of Haar-random states is known to be extensive ~\cite{ManyBodymagic,Turkeshi_2025}. In contrast, the cut-summed bipartite response remains $O(1)$, showing that total BNMR residing in bipartite entanglement remains bounded in the thermodynamic limit. This raises the open question of whether nonlocal magic resource carried by multipartite entanglement can exhibit extensive $O(N)$ scaling in Haar-random states.

\subsection{Canonical encoding of Schmidt rank-2 states}
We now specialize Eq.~\eqref{eq:final_chi} to pure many-body states with Schmidt rank $2$ across a fixed bipartition. Any such state admits the Schmidt decomposition
\begin{equation}
\label{eq:rank2_schmidt}
|\Psi\rangle_{AB}
=
\cos\theta\,|u_1\rangle_A\otimes|v_1\rangle_B
+
\sin\theta\,|u_2\rangle_A\otimes|v_2\rangle_B,
\end{equation}
where $\theta\in[0,\pi/2]$, so that $\cos\theta$ and $\sin\theta$ are the Schmidt coefficients, and
$\{|u_1\rangle_A,|u_2\rangle_A\}$ and $\{|v_1\rangle_B,|v_2\rangle_B\}$
form orthonormal Schmidt bases for subsystems $A$ and $B$, respectively. We implement the canonical encoding through local unitary operations $V_A$ and $V_B$ acting on the Schmidt bases as
\[
V_A |u_1\rangle_A
=
|0\rangle_{A_{\rm code}} \otimes |\overline{0}\rangle_{A_{\rm anc}},
\qquad
V_A |u_2\rangle_A
=
|1\rangle_{A_{\rm code}} \otimes |\overline{0}\rangle_{A_{\rm anc}},
\]
\[
V_B |v_1\rangle_B
=
|0\rangle_{B_{\rm code}} \otimes |\overline{0}\rangle_{B_{\rm anc}},
\qquad
V_B |v_2\rangle_B
=
|1\rangle_{B_{\rm code}} \otimes |\overline{0}\rangle_{B_{\rm anc}}.
\]
Applying $V_A \otimes V_B$ to Eq.~\eqref{eq:rank2_schmidt}, we obtain
\[
(V_A \otimes V_B)|\Psi\rangle_{AB}
=
\left(
\cos\theta\,|00\rangle
+
\sin\theta\,|11\rangle
\right)
\otimes
|\overline{0}\rangle_{A_{\rm anc}}
\otimes
|\overline{0}\rangle_{B_{\rm anc}}.
\]
Thus, the minimal encoding core reduces to a pure entangled qubit pair. Without loss of generality, local SWAP operations place the two encoded qubits on adjacent sites across the cut, as illustrated in Fig.~\ref{fig:Bell}.
By Eq.~\eqref{eq:final_chi}, the evaluation of the BNMR therefore reduces to
\[
M_{\mathrm{bp}}(|\Psi\rangle_{AB})
=
M_{\mathrm{bp}}
\!\left(
\cos\theta\,|00\rangle+\sin\theta\,|11\rangle
\right).
\]

\begin{figure}[htbp]
\centering
\includegraphics[scale=0.5]{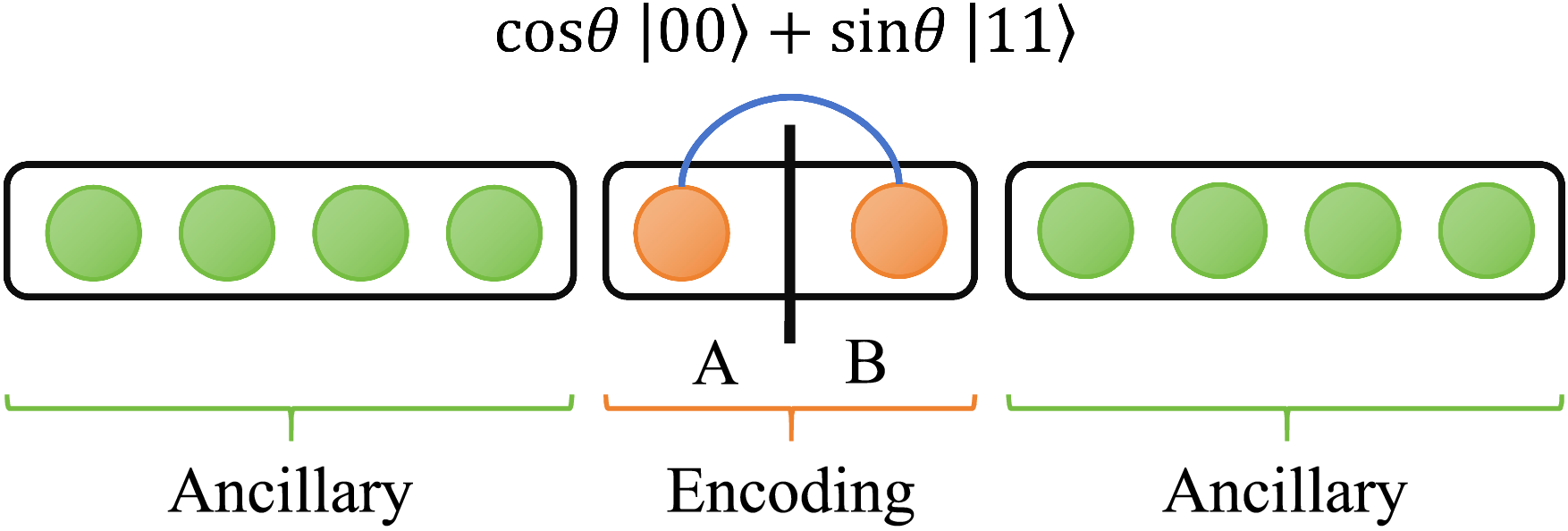}
\caption{Canonical representative of a pure many-qubit state with Schmidt rank $2$ across a fixed bipartition. Under local-unitary transformations, the nontrivial bipartite structure is encoded into a pure entangled qubit pair across the cut, which determines the BNMR.}
\label{fig:Bell}
\end{figure}

Denoting the encoded pair by $|\phi_\theta\rangle_{ab} =\cos\theta\,|00\rangle+\sin\theta\,|11\rangle$, the BNMR of the original many-body state is fully determined by $|\phi_\theta\rangle_{ab}$. Qian and Wang demonstrated that the magic resource of this two-qubit state is entirely nonlocal~\cite{BNM}. Its BNMR  $M_{\mathrm{bp}}(|\phi_\theta\rangle_{ab})$ therefore coincides with its total magic resource $M_2(|\phi_\theta\rangle_{ab})$, yielding the closed-form expression for any many-body pure state with Schmidt rank $2$:
\begin{equation} \label{eq:rank2_final}
M_{\mathrm{bp}}(|\Psi\rangle_{AB})  = M_2(|\phi_\theta\rangle_{ab}) = \ln\!\left(\frac{8}{7+\cos 8\theta}\right).
\end{equation}

It is worth emphasizing that, although the encoded pair is placed on neighboring qubits across the cut in Fig.~\ref{fig:Bell}, this geometric placement does not imply that the BNMR of the original state is restricted to nearest-neighbor correlations. Rather, $M_{\mathrm{bp}}(|\Psi\rangle_{AB})$ is encoded in the correlations between subsystems $A$ and $B$, which allows both short-range and long-range entanglement. In the following, we use Eq.~\eqref{eq:rank2_final} to analyze the distribution and hierarchy of nonlocal magic resource in genuinely multipartite entangled states.

\subsection{Hierarchy collapse in generalized GHZ states}
We now consider the generalized $N$-qubit GHZ state
\begin{equation}
\label{eq:GHZ_state}
|\mathrm{GHZ}_N(\theta)\rangle
=
\cos\theta\,|0\rangle^{\otimes N}
+
\sin\theta\,|1\rangle^{\otimes N},
\qquad
\theta\in(0,\pi/2).
\end{equation}
This state is genuinely multipartite entangled, meaning that it is entangled across every nontrivial bipartition. For an arbitrary bipartition with $n_A+n_B=N$, it can be written as
\begin{equation}
\label{eq:GHZ_bipartition}
|\mathrm{GHZ}_N(\theta)\rangle
=
\cos\theta\,|0\rangle_A^{\otimes n_A}\otimes|0\rangle_B^{\otimes n_B}
+
\sin\theta\,|1\rangle_A^{\otimes n_A}\otimes|1\rangle_B^{\otimes n_B}.
\end{equation}
Since the two states on each subsystem are orthogonal, Eq.~\eqref{eq:GHZ_bipartition} is already the Schmidt decomposition for that bipartition. Therefore, for every cut the Schmidt rank is always $2$, with Schmidt coefficients $\cos\theta$ and $\sin\theta$. It follows from Eq.~\eqref{eq:rank2_final} that the BNMR takes the same value for every bipartition:
\begin{equation}
\label{eq:GHZ_bipartite_magic}
M_{\mathrm{bp}}\!\left(|\mathrm{GHZ}_N(\theta)\rangle\right)
=
\ln\!\left(\frac{8}{7+\cos 8\theta}\right).
\end{equation}

The nested optimization domains constrains the global nonlocal magic resource $M_{\mathrm{NL}}$. By definition, $M_{\mathrm{NL}}$ is minimized over fully local product unitary operations $\bigotimes_{i=1}^{N} U_i$. Since this set trivially contains the identity operator, $M_{\mathrm{NL}}$ is upper-bounded by the total magic resource $M$. Conversely, the fully local product unitary operations form a subset of the bipartite local unitary operations $U_A\otimes U_B$, $M_{\mathrm{NL}}$ is lower-bounded by $M_{\mathrm{bp}}$. Therfore,
\begin{equation}
\label{eq:hierarchy}
M_{\mathrm{bp}}\!\left(|\mathrm{GHZ}_N(\theta)\rangle\right)\le M_{\mathrm{NL}}\!\left(|\mathrm{GHZ}_N(\theta)\rangle\right)\le M\!\left(|\mathrm{GHZ}_N(\theta)\rangle\right).
\end{equation}

The generalized GHZ state can be constructed from the pure entangled qubit pair $|\phi_\theta\rangle_{ab} = \cos\theta\,|00\rangle + \sin\theta\,|11\rangle$ by appending the $(N-2)$-qubit stabilizer ancilla $|0\rangle^{\otimes (N-2)}$ and sequentially applying CNOT gates to spread the correlation across the system~\cite{ManyBodymagic}. Because the SRE is invariant under Clifford unitary operations and appending stabilizer ancillas, the total magic resource is preserved. Thus,
\begin{equation}
\label{eq:GHZ_SRE_from_pair}
M\!\left(|\mathrm{GHZ}_N(\theta)\rangle\right)
=
M\!\left(|\phi_\theta\rangle_{ab}\right)
=
\ln\!\left(\frac{8}{7+\cos 8\theta}\right).
\end{equation}

Notably, Eq.~\eqref{eq:GHZ_bipartite_magic}, \eqref{eq:hierarchy}, and \eqref{eq:GHZ_SRE_from_pair} show that the lower and upper bounds on $M_{\mathrm{NL}}$ coincide. Hence,
\begin{equation}
\label{eq:GHZ_magic_collapse}
M_{\mathrm{bp}}\!\left(|\mathrm{GHZ}_N(\theta)\rangle\right)
=
M_{\mathrm{NL}}\!\left(|\mathrm{GHZ}_N(\theta)\rangle\right)
=
M\!\left(|\mathrm{GHZ}_N(\theta)\rangle\right).
\end{equation}
This result reveals an intriguing hierarchy collapse for generalized GHZ states. The BNMR extracted from any bipartition already equals the global nonlocal magic resource and the total magic resource of the entire state. This is nontrivial because the correlations of a generalized GHZ state are intrinsically $N$-partite. Tracing out any single qubit can destroy all entanglement and magic resource of the system, leaving the remaining state as a separable classical mixture of stabilizer states. One might therefore expect the global nonlocal magic resource to contain an irreducible multipartite component that cannot be recovered from a single bipartition. Instead, the result shows that the entire nonlocal magic resource is encoded in a rank-$2$ global logical structure, which is fully resolved by every bipartition.

\begin{figure}[htbp]
\centering
\includegraphics[scale=0.5]{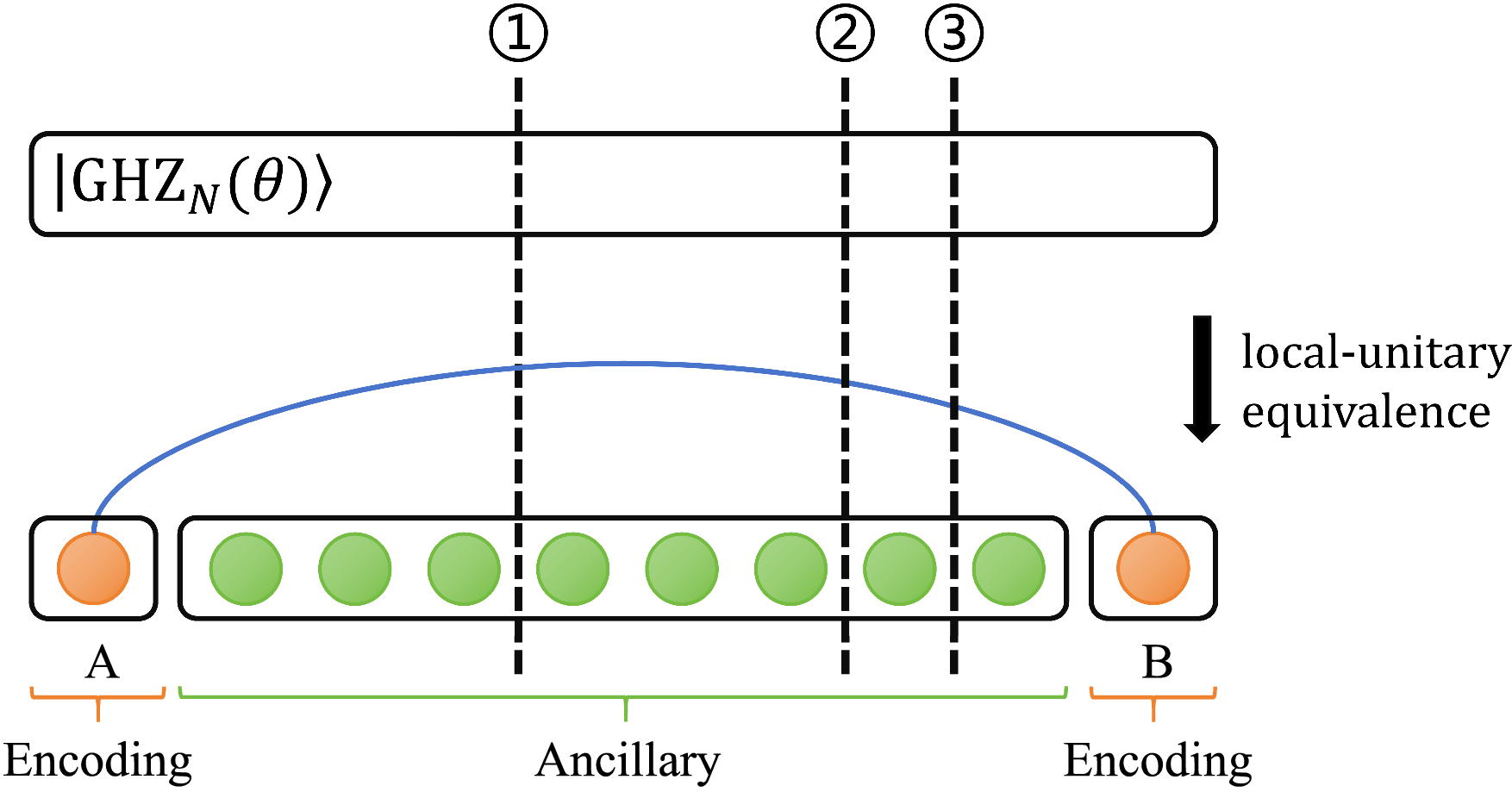}
\caption{End-to-end canonical encoding of the generalized GHZ state. Local SWAP operations move the two encoded qubits to the outermost sites without changing the BNMR. Every contiguous bipartition resolves the same rank-$2$ logical entanglement structure, leading to the collapse $M_{\mathrm{bp}}=M_{\mathrm{NL}}$.}
\label{fig:GHZ}
\end{figure}

The hierarchy collapse also admits an intuitive interpretation from the canonical encoding perspective. In Fig.~\ref{fig:Bell}, the minimal encoding core was placed on two adjacent qubits across the cut. Since SWAP gates acting within each subsystem are local Clifford operations, the encoded qubits can instead be moved to the two ends of the chain, as illustrated in Fig.~\ref{fig:GHZ}. In this representation, the generalized GHZ state can be viewed as a single end-to-end rank-$2$ logical structure, that is, a long-range entangled qubit pair carrying the entire global nonlocal magic resource. Every contiguous cut intersects the same structure and therefore extracts the same BNMR. Because the generalized GHZ state is permutation symmetric, the same picture can extend to arbitrary bipartitions.

\section{Conclusion}

In conclusion, we have established an intrinsic spectral framework for characterizing bipartite nonlocal magic resource (BNMR) in pure quantum states. By introducing a canonical encoding scheme, we demonstrate that BNMR can be exactly compressed into a minimal Schmidt-support core. This dimension reduction proves that pure-state BNMR is an explicit function of the nonzero Schmidt spectrum, extending its invariance from local unitary transformations to local isometries. Consequently, the entanglement spectrum is revealed not merely as a descriptor of bipartite entanglement, but as the fundamental structure dictating irreducible nonstabilizerness.

Adopting this spectral universality, we derived the exact leading-order response of BNMR near its zero-resource set, proving that deviations from stabilizer-compatible flat spectra generate BNMR quadratically with a fixed leading coefficient. Applying this perturbative scaling to Haar-random states yields an analytic BNMR profile that is sharply localized at the symmetric bipartition and exponentially suppressed toward asymmetric cuts. This localized behavior stands in stark contrast to the broadening Page curve of entanglement. Furthermore, our closed-form analysis of Schmidt rank-2 states uncovers a striking hierarchy collapse in generalized GHZ states, implying that global nonlocal magic resource can be entirely encoded within a minimal logical entanglement structure.

Consequently, our results identify the Schmidt spectrum as the precise microscopic bridge between entanglement and nonlocal magic resource. This perspective indicates that the distribution of magic resource in many-body systems is governed by the concrete spectral geometry of entanglement. Our framework provides a concise route for analyzing BNMR in structured many-body states and paves the way for understanding how bipartite and genuinely multipartite nonlocal magic resources jointly organize quantum complexity.

\begin{acknowledgments}
The authors gratefully acknowledge support from the National Natural Science Foundation of China (Grant No.~12374107).
\end{acknowledgments}

\appendix
\section{Proof of Lemma 2} \label{sec:proof-lemma-2}

Fix a stabilizer state $|s_0\rangle$. Any normalized pure state $|\psi\rangle$ sufficiently close to $|s_0\rangle$ can be written as
\begin{equation}
    |\psi\rangle = \sqrt{1-t}\,|s_0\rangle + \sqrt{t}\,|\phi\rangle, \qquad \langle s_0|\phi\rangle = 0,
    \label{eq:SM-local-decomposition}
\end{equation}
where
\begin{equation}
    t = 1 - |\langle s_0|\psi\rangle|^2 \ll 1.
    \label{eq:SM-t-definition}
\end{equation}
Define the phase-free Pauli support of $|s_0\rangle$ by
\begin{equation}
    \mathcal{S}_{s_0} = \left\{ P\in\mathcal{P}_N : \langle s_0|P|s_0\rangle = \pm 1 \right\}, \qquad \sigma_P = \langle s_0|P|s_0\rangle.
    \label{eq:SM-Pauli-support}
\end{equation}
For $P\in\mathcal{S}_{s_0}$, one has $P|s_0\rangle = \sigma_P|s_0\rangle$ with $\sigma_P = \pm 1$. The cross term in $\langle\psi|P|\psi\rangle$ vanishes since
\begin{equation}
    \langle s_0|P|\phi\rangle = \sigma_P\langle s_0|\phi\rangle = 0.
\end{equation}
Consequently, we have
\begin{equation}
    \langle\psi|P|\psi\rangle = \sigma_P(1-t) + t\langle\phi|P|\phi\rangle, \qquad P\in\mathcal{S}_{s_0},
    \label{eq:SM-P-expectation-support}
\end{equation}
and hence
\begin{equation}
    \langle\psi|P|\psi\rangle^4 = 1 + 4t\left(\sigma_P\langle\phi|P|\phi\rangle - 1\right) + O(t^2).
    \label{eq:SM-fourth-power-support}
\end{equation}

To evaluate the sum over $\mathcal{S}_{s_0}$, let $G_1,\ldots,G_N$ be a set of independent stabilizer generators of $|s_0\rangle$, satisfying
\begin{equation}
    G_j|s_0\rangle = |s_0\rangle, \qquad j=1,\ldots,N.
\end{equation}
The projector onto their common $+1$ eigenspace is
\begin{equation}
    \Pi_{s_0} = \prod_{j=1}^{N}\frac{I+G_j}{2} = \frac{1}{2^N} \sum_{a_1,\ldots,a_N\in\{0,1\}} G_1^{a_1}\cdots G_N^{a_N}.
    \label{eq:SM-stabilizer-projector-product}
\end{equation}
Since $|s_0\rangle$ is a pure stabilizer state, the common eigenspace is one-dimensional, yielding
\begin{equation}
    \Pi_{s_0} = |s_0\rangle\langle s_0|.
    \label{eq:SM-projector-rank-one}
\end{equation}
Equivalently, using the phase-free Pauli support $\mathcal{S}_{s_0}$ defined above, each stabilizer-group element can be written uniquely as $\sigma_P P$ with $P\in\mathcal{S}_{s_0}$. Therefore,
\begin{equation}
    2^{-N}\sum_{P\in\mathcal{S}_{s_0}}\sigma_P P = \Pi_{s_0} = |s_0\rangle\langle s_0|, \qquad |\mathcal{S}_{s_0}| = 2^N.
    \label{eq:SM-stabilizer-projector}
\end{equation}
Using this relation, it follows that
\begin{equation}
    \sum_{P\in\mathcal{S}_{s_0}} \sigma_P\langle\phi|P|\phi\rangle = 2^N|\langle s_0|\phi\rangle|^2 = 0,
\end{equation}
and therefore
\begin{equation}
    \sum_{P\in\mathcal{S}_{s_0}} \langle\psi|P|\psi\rangle^4 = 2^N(1-4t) + O(t^2).
    \label{eq:SM-support-sum}
\end{equation}
For $P\notin\mathcal{S}_{s_0}$, a stabilizer state has $\langle s_0|P|s_0\rangle = 0$. Equation~\eqref{eq:SM-local-decomposition} then gives
\begin{equation}
    \langle\psi|P|\psi\rangle = 2\sqrt{t(1-t)}\,\operatorname{Re}(\langle s_0|P|\phi\rangle) + t\langle\phi|P|\phi\rangle = O(\sqrt{t}),
\end{equation}
so that $\langle\psi|P|\psi\rangle^4 = O(t^2)$. Since $N$ is fixed, $\mathcal{P}_N$ is finite, and the total contribution from $P\notin\mathcal{S}_{s_0}$ remains $O(t^2)$. Combining the two sectors yields
\begin{equation}
    2^{-N}\sum_{P\in\mathcal{P}_N} \langle\psi|P|\psi\rangle^4 = 1 - 4t + O(t^2).
    \label{eq:SM-Pauli-moment-expansion}
\end{equation}
Substituting Eq.~\eqref{eq:SM-Pauli-moment-expansion} into the definition of the stabilizer 2-R\'{e}nyi entropy gives
\begin{equation}
    M_2(|\psi\rangle) = -\ln\left[1 - 4t + O(t^2)\right].
\end{equation}
Expanding the logarithm around unity, we obtain
\begin{equation}
    M_2(|\psi\rangle) = 4t + O(t^2) = 4\left(1 - |\langle s_0|\psi\rangle|^2\right) + O\bigl( (1 - |\langle s_0|\psi\rangle|^2)^2 \bigr),
    \label{eq:SM-lemma-2-local-expansion}
\end{equation}
which proves the first statement of Lemma~2.

For a fixed number of qubits, there are only finitely many physically distinct pure stabilizer states. Therefore, the local expansion in Eq.~\eqref{eq:SM-lemma-2-local-expansion} can be made uniform over the stabilizer set. Now consider a pure state $|\psi\rangle$ lying in this uniform neighborhood of the stabilizer set. Let $|\tau_*\rangle\in\mathrm{Stab}$ be a closest stabilizer state in fidelity, namely,
\begin{equation}
    |\langle\tau_*|\psi\rangle|^2 = \max_{|\tau\rangle\in\mathrm{Stab}} |\langle\tau|\psi\rangle|^2.
\end{equation}
Then, we define $\delta_{\mathrm{stab}}(|\psi\rangle) = 1 - |\langle\tau_*|\psi\rangle|^2$. Applying the uniform expansion around $|\tau_*\rangle$ gives, for some constant $C_1>0$,
\begin{equation}
    M_2(|\psi\rangle) \geq 4\,\delta_{\mathrm{stab}}(|\psi\rangle) - C_1\,\delta_{\mathrm{stab}}(|\psi\rangle)^2.
    \label{eq:SM-uniform-stabilizer-lower-bound}
\end{equation}
This proves the second statement of Lemma~2.

\section{Proof of Lemma 3}
\label{sec:proof-lemma-3}

We first recall a standard overlap bound for bipartite pure states. Let $|\Phi\rangle$ and $|\Psi\rangle$ be two bipartite pure states with Schmidt spectra $\boldsymbol{\alpha}$ and $\boldsymbol{\beta}$, respectively. Under the state-matrix correspondence, the singular values of their coefficient matrices are $\{\sqrt{\alpha_i}\}$ and $\{\sqrt{\beta_i}\}$. The von Neumann trace inequality~\cite{Horn_Johnson_1985} then gives
\begin{equation}
    |\langle\Phi|\Psi\rangle| \leq \sum_i \sqrt{\alpha_i^{\downarrow}\beta_i^{\downarrow}},
    \label{eq:SM-Schmidt-overlap-bound}
\end{equation}
where $\downarrow$ denotes nonincreasing order, and zero padding is understood when the Schmidt ranks are different.

We now apply this bound to the state $|\Psi\rangle_{AB}$ in Lemma~3. Its Schmidt spectrum is $\boldsymbol{\lambda} = \boldsymbol{u}_r + \boldsymbol{\eta}$, with Schmidt rank $r=2^n$. Let $|\tau\rangle\in\mathrm{Stab}$ be an arbitrary bipartite stabilizer state, and denote its Schmidt rank under the same fixed bipartition by $q$. The bipartite normal form of stabilizer states implies that $q=2^m$ for some integer $m\geq 0$, and that the nonzero Schmidt spectrum of $|\tau\rangle$ is the flat spectrum
\begin{equation}
    \boldsymbol{u}_q = \left( \frac{1}{q},\ldots,\frac{1}{q} \right).
\end{equation}
Therefore, Eq.~\eqref{eq:SM-Schmidt-overlap-bound} gives
\begin{equation}
    |\langle\tau|\Psi\rangle_{AB}|^2 \leq F_q(\boldsymbol{\lambda}),
    \label{eq:SM-rank-q-fidelity-bound}
\end{equation}
where
\begin{equation}
    F_q(\boldsymbol{\lambda}) = \frac{1}{q} \Biggl( \sum_{i=1}^{\min\{q,r\}} \sqrt{\lambda_i^{\downarrow}} \Biggr)^2 .
    \label{eq:SM-Fq-definition}
\end{equation}
Here $F_q(\boldsymbol{\lambda})$ denotes the Schmidt-spectrum upper bound on the squared overlap between $|\Psi\rangle_{AB}$ with Schmidt spectrum $\boldsymbol{\lambda}$, and any stabilizer state with Schmidt rank $q$. The subscript $q$ labels the Schmidt rank of the candidate stabilizer state, while the input $\boldsymbol{\lambda}$ is the Schmidt spectrum of the fixed state $|\Psi\rangle_{AB}$.

The reason for keeping $q$ arbitrary is that the stabilizer set contains states with different Schmidt ranks. To lower bound the distance from $|\Psi\rangle_{AB}$ to the entire stabilizer set, we must exclude the possibility that a stabilizer state with $q\neq r$ has a larger fidelity with $|\Psi\rangle_{AB}$ than the same-rank flat sector with Schmidt rank $r$.

This exclusion follows by comparing the two flat spectra $\boldsymbol{u}_r$ and $\boldsymbol{u}_q$. Setting $\boldsymbol{\lambda}=\boldsymbol{u}_r$ in Eq.~\eqref{eq:SM-Fq-definition}, we obtain
\begin{equation}
    F_q(\boldsymbol{u}_r) =
    \begin{cases}
        q/r, & q<r, \\
        1,   & q=r, \\
        r/q, & q>r.
    \end{cases}
    \label{eq:SM-flat-spectrum-comparison}
\end{equation}
Since both $q$ and $r$ are powers of two, any $q\neq r$ differs from $r$ by at least a factor of two. Hence
\begin{equation}
    F_q(\boldsymbol{u}_r) \leq \frac{1}{2}, \qquad q\neq r,
\end{equation}
whereas
\begin{equation}
    F_r(\boldsymbol{u}_r) = 1.
\end{equation}
Thus, at the flat spectrum $\boldsymbol{u}_r$, the same-rank flat sector with Schmidt rank $r$ is separated from all different-rank flat sectors by a finite fidelity gap.

Since the Hilbert space is finite dimensional, only finitely many stabilizer Schmidt ranks $q$ can occur. By continuity, this finite gap persists in a sufficiently small neighborhood of $\boldsymbol{u}_r$. Therefore, there exist such a neighborhood and a constant $c_0<1$ such that, for all $\boldsymbol{\lambda}$ in this neighborhood,
\begin{equation}
    F_q(\boldsymbol{\lambda}) \leq c_0 < F_r(\boldsymbol{\lambda}), \qquad q\neq r.
    \label{eq:SM-same-rank-dominates}
\end{equation}
Consequently, for any stabilizer state $|\tau\rangle$, regardless of its Schmidt rank $q$, Eq.~\eqref{eq:SM-rank-q-fidelity-bound} implies
\begin{equation}
    |\langle\tau|\Psi\rangle_{AB}|^2 \leq F_q(\boldsymbol{\lambda}) \leq F_r(\boldsymbol{\lambda}).
\end{equation}
Taking the maximum over all stabilizer states gives
\begin{equation}
    \max_{|\tau\rangle\in\mathrm{Stab}} |\langle\tau|\Psi\rangle_{AB}|^2 \leq F_r(\boldsymbol{\lambda}).
    \label{eq:SM-max-fidelity-Fr}
\end{equation}
Therefore,
\begin{equation}
    \delta_{\mathrm{stab}}(|\Psi\rangle_{AB}) = 1 - \max_{|\tau\rangle\in\mathrm{Stab}} |\langle\tau|\Psi\rangle_{AB}|^2 \geq 1 - F_r(\boldsymbol{\lambda}).
    \label{eq:SM-delta-lower-Fr}
\end{equation}

It remains to expand $F_r(\boldsymbol{\lambda})$ for $\boldsymbol{\lambda}=\boldsymbol{u}_r+\boldsymbol{\eta}$. Since $\lambda_i=1/r+\eta_i$ and $\|\boldsymbol{\eta}\|_2\ll 1$, we have
\begin{equation}
    \sqrt{\lambda_i} = \frac{1}{\sqrt{r}} + \frac{\sqrt{r}}{2}\eta_i - \frac{r^{3/2}}{8}\eta_i^2 + O(|\eta_i|^3).
    \label{eq:SM-sqrt-expansion}
\end{equation}
Using $\sum_{i=1}^r\eta_i=0$, we obtain
\begin{equation}
    \sum_{i=1}^r\sqrt{\lambda_i} = \sqrt{r} - \frac{r^{3/2}}{8}\|\boldsymbol{\eta}\|_2^2 + O(\|\boldsymbol{\eta}\|_2^3).
    \label{eq:SM-sqrt-sum}
\end{equation}
Therefore,
\begin{align}
    F_r(\boldsymbol{\lambda})
    &= \frac{1}{r} \Biggl( \sum_{i=1}^r\sqrt{\lambda_i} \Biggr)^2 \notag\\
    &= \frac{1}{r} \biggl( \sqrt{r} - \frac{r^{3/2}}{8}\|\boldsymbol{\eta}\|_2^2 + O(\|\boldsymbol{\eta}\|_2^3) \biggr)^2 \notag\\
    &= 1 - \frac{r}{4}\|\boldsymbol{\eta}\|_2^2 + O(\|\boldsymbol{\eta}\|_2^3).
    \label{eq:SM-Fr-expansion}
\end{align}
Hence
\begin{equation}
    1 - F_r(\boldsymbol{\lambda}) = \frac{r}{4}\|\boldsymbol{\eta}\|_2^2 + O(\|\boldsymbol{\eta}\|_2^3).
\end{equation}
Equivalently, for some constant $C_2>0$ and all sufficiently small $\|\boldsymbol{\eta}\|_2$,
\begin{equation}
    1 - F_r(\boldsymbol{\lambda}) \geq \frac{r}{4}\|\boldsymbol{\eta}\|_2^2 - C_2\|\boldsymbol{\eta}\|_2^3.
    \label{eq:SM-one-minus-Fr}
\end{equation}
Combining Eqs.~\eqref{eq:SM-delta-lower-Fr} and \eqref{eq:SM-one-minus-Fr}, we obtain
\begin{equation}
    \delta_{\mathrm{stab}}(|\Psi\rangle_{AB}) \geq \frac{r}{4}\|\boldsymbol{\eta}\|_2^2 - C_2\|\boldsymbol{\eta}\|_2^3,
\end{equation}
which proves Lemma~3.

\bibliography{BNMR}

\end{document}